\documentclass[11pt]{article}
\usepackage[utf8]{inputenc}
\usepackage[margin=1in]{geometry}

\usepackage{graphicx}
\usepackage{xcolor}
\usepackage{amsmath}
\usepackage{amssymb}
\usepackage{amsthm}
\usepackage{algorithm}
\usepackage{algpseudocode}
\usepackage{hyperref}
\hypersetup{
     colorlinks = true,
     linkcolor = blue,
     citecolor = blue,
     filecolor = blue,
     urlcolor = blue
     }
\usepackage{cleveref}
\usepackage{framed}

\algnotext{EndFor}
\algnotext{EndIf}
\algnotext{EndWhile}
\algnotext{EndProcedure}

\allowdisplaybreaks

\newtheorem{observation}{Observation}
\newtheorem{lemma}{Lemma}
\newtheorem{theorem}[lemma]{Theorem}

\newtheorem{definition}[lemma]{Definition}

\newtheorem{assumption}[lemma]{Assumption}

\newtheorem{corollary}[lemma]{Corollary}
\newtheorem{question}[lemma]{Question}

\newcommand{\eps}{\epsilon}
\newcommand{\E}{\mathbb{E}}

\newcommand{\A}{\mathbb{A}}

\newcommand{\cQ}{\mathcal{Q}}
\newcommand{\cP}{\mathcal{P}}

\DeclareMathOperator{\argmax}{argmax}

\newcommand{\ind}{\mathbf{1}}

\newcommand{\cT}{\mathcal{T}}

\newcommand{\C}{\mathcal{C}}
\newcommand{\cH}{\mathcal{H}}
\newcommand{\cG}{\mathcal{G}}
\newcommand{\N}{\mathbb{N}}

\newcommand{\OPT}{\mathrm{OPT}}
\newcommand{\VC}{\mathrm{VC}}
\newcommand{\Ldim}{\mathrm{Ldim}}
\newcommand{\SAT}{\mathrm{SAT}}

\definecolor{Gred}{RGB}{219, 50, 54}
\definecolor{Ggreen}{RGB}{60, 186, 84}
\definecolor{Gblue}{RGB}{72, 133, 237}
\definecolor{Gyellow}{RGB}{247, 178, 16}
\definecolor{ToCgreen}{RGB}{0, 128, 0}
\definecolor{myGold}{RGB}{231,141,20}
\definecolor{myBlue}{rgb}{0.19,0.41,.65}
\definecolor{myPurple}{RGB}{175,0,124}

\usepackage[colorinlistoftodos,prependcaption,textsize=scriptsize]{todonotes}
\providecommand{\Comments}{3}
\ifnum\Comments=2
\paperwidth=\dimexpr \paperwidth + 2.8cm\relax
\evensidemargin=\dimexpr\evensidemargin + 2.2cm\relax
\fi
\newcommand{\mytodo}[1]{\ifnum\Comments=1{#1}\fi}

\newcommand{\tableoftodos}{\ifnum\Comments=1 \listoftodos[Comments/To Do's] \fi}

\title{Improved Inapproximability of VC Dimension and \\ Littlestone's Dimension via (Unbalanced) Biclique}

\author{
\makebox[.3\linewidth]
{Pasin Manurangsi}\\
Google Research, Thailand\\
\texttt{pasin@google.com}
}

\date{\today}
\begin{document}

\maketitle

\begin{abstract}
We study the complexity of computing (and approximating) VC Dimension and Littlestone's Dimension when we are given the concept class explicitly. 
We give a simple reduction from Maximum (Unbalanced) Biclique problem to approximating  VC Dimension and Littlestone's Dimension.
With this connection, we derive a range of hardness of approximation results and running time lower bounds. For example, under the (randomized) Gap-Exponential Time Hypothesis or the Strongish Planted Clique Hypothesis, we show a tight inapproximability result: both dimensions are hard to approximate to within a factor of $o(\log n)$ in polynomial-time. These improve upon constant-factor inapproximability results from~\cite{ManurangsiR17}.
\end{abstract}

\section{Introduction}

VC Dimension~\cite{VC} and Littlestone's Dimension~\cite{Littlestone87} are two of the most fundamental quantities in learning theory; the former governs the sample complexity in PAC model to within a constant factor~\cite{BlumerEHW89,Hanneke16}, while the latter governs the sample complexity in online learning~\cite{Littlestone87}. Due to this, it would be extremely useful to have an efficient algorithm for computing or approximating these dimensions for any given concept class $\C \subseteq 2^X$.

Given the importance of the two quantities, it should come as no surprise that this question has been investigated for several decades. Two models have been considered, based on whether the concept class is given explicitly. In the first ``implicit'' model, the input is a circuit which when given (indices of) $C \in \C$ and $x \in X$, evaluates to $C(x) \in \{0, 1\}$. For this model, both problems are hard: Schaefer proved that computing VC Dimension is $\Sigma_p^3$-complete~\cite{Schaefer99} whereas Littlestone's Dimension is PSPACE-complete~\cite{Schaefer00}. Furthermore, Mossel and Umans~\cite{MosselU02} showed that VC Dimension is $\Sigma_p^3$-hard even to approximate to within a factor less than 2, and that it is AM-hard to approximate to within $n^{1 - \eps}$-factor for any constant $\eps > 0$.

The second model--which is the focus of the remainder of this work--is the ``explicit'' model, where the concept class is given as an $|X| \times |\C|$-matrix. It is not hard to see that straightforward algorithms can solve both problems in $n^{O(\log n)}$-time, because both dimensions are at most $O(\log n)$. Therefore, these problems cannot be NP-hard (unless NP is contained in $DTIME(n^{O(\log n)}$)). Papadimitriou and Yannakakis~\cite{PapadimitriouY96} defined a class similar to NP but with ``limited non-determinism'' called LOGNP, and showed that VC Dimension is complete for this class. Frances and Litman~\cite{FrancesL98} showed that Littlestone's Dimension is also LOGNP-hard. A consequence of these results is that, assuming the Exponential Time Hypothesis (ETH)\footnote{ETH~\cite{ImpagliazzoP01,ImpagliazzoPZ01} states that 3SAT cannot be solved in $2^{o(N)}$ time where $N$ denotes the number of variables.}, VC Dimension and Littlestone's Dimension cannot be solved in $n^{o(\log n)}$ time. More recently, Manurangsi and Rubinstein~\cite{Manurangsi17} extended these lower bounds to rule out even approximation algorithms. Specifically, they showed, assuming ETH, that there is no $(2 - o(1))$-approximation algorithm for VC Dimension and no $(1 + \eps)$-approximation algorithm for Littlestone's Dimension that runs in $n^{\log^{1 - o(1)} n}$ time. Nonetheless, it remains open whether any constant factor approximation is achievable in polynomial time:

\begin{question} \label{openq}
Is there polynomial-time $O(1)$-approximation algorithm for VC Dimension or Littlestone's Dimension?
\end{question}

As mentioned earlier, VC Dimension characterizes the sample complexity in PAC learning only to within a (large) constant factor~\cite{BlumerEHW89,Hanneke16}. Therefore, \Cref{openq} is not only a natural question but it can also be argued that $O(1)$-approximation of VC Dimension is essentially as good as computing it exactly for the purpose of approximating the sample complexity.


\subsection{Our Contributions}
\label{sec:our-contrib}

We answer the above question negatively, under certain computational complexity assumptions. Before we state our results, let us briefly recall the assumptions that we rely on. (All assumptions are formalized in \Cref{sec:prelim}.) In addition to ETH, we will use its strengthening called \emph{Gap-ETH}~\cite{Dinur16,ManurangsiR17-gap-eth} which (roughly) says that there is no $2^{o(N)}$-time algorithm even for an approximate version of 3SAT. The \emph{Planted Clique Hypothesis}~\cite{Karp76,Jerrum92} states that there is no polynomial-time algorithm that can distinguish between an Erdos-Renyi random graph and one with a planted clique of size $N^{\Omega(1)}$. The \emph{Strongish Planted Clique Hypothesis}~\cite{ManurangsiRS21} is a strengthening of Planted Clique Hypothesis that rules out even $N^{o(\log N)}$-time algorithms.

With all the assumptions in mind, we can now state our results. We divide the results into two groups, based on its inapproximability ratio. In the first group of results, we show, under ETH, Gap-ETH and Strongish Planted Clique Hypothesis, that there is no polynomial-time $\tilde{o}(\log n)$-approximation algorithm for VC Dimension / Littlestone's Dimension, as stated below.

\begin{theorem} \label{thm:eth-inapprox}
Assuming ETH, there is no polynomial-time algorithm for approximating VC Dimension or Littlestone's Dimension to within $o\left(\frac{\log n}{(\log \log n)^{\xi}}\right)$ factor for some constant $\xi > 0$.
\end{theorem}

\begin{theorem} \label{thm:gap-eth-inapprox}
Assuming Gap-ETH, there is no polynomial-time algorithm for approximating VC Dimension or Littlestone's Dimension to within $o(\log n)$ factor.
\end{theorem}

\begin{theorem} \label{thm:spc-approx}
Assuming the Strongish Planted Clique Hypothesis, there is no polynomial-time algorithm for approximating VC Dimension or Littlestone's Dimension to within $o(\log n)$ factor.
\end{theorem}

Not only do these results answer~\Cref{openq}, but the inapproximability factors are essentially the best possible: an algorithm that always output one is an $(\log n)$-approximation algorithm.


If we only want to rule out any constant-factor approximation, then our lower bounds hold not only against polynomial-time algorithms but slightly quasi-polynomial time algorithms as well:

\begin{theorem} \label{thm:spc-time}
Assuming the Strongish Planted Clique Hypothesis, there is no $n^{o(\log n)}$-time algorithm for approximating VC Dimension or Littlestone's Dimension to within any constant factor.
\end{theorem}

\begin{theorem} \label{thm:gap-eth-time}
Assuming Gap-ETH, there is no $n^{o\left({(\log n)^{1/3}}\right)}$-time algorithm for approximating VC Dimension or Littlestone's Dimension to within any constant factor.
\end{theorem}

\begin{theorem} \label{thm:eth-time}
Assuming ETH, there is no $n^{o\left(\frac{(\log n)^{1/3}}{(\log \log n)^{\xi}}\right)}$-time algorithm for approximating VC Dimension or Littlestone's Dimension to within any constant factor for some constant $\xi > 0$.
\end{theorem}

Recall that there is also an $n^{O(\log n)}$-time exact algorithm for computing both dimensions. Therefore, the running time lower bound in \Cref{thm:spc-time} is tight. On the other hand, while \Cref{thm:gap-eth-time,thm:eth-time} achieve larger inapproximability factors compared to \cite{ManurangsiR17}, the running time lower bounds are weaker than that in \cite{ManurangsiR17} due to technical reasons which we discuss more in \Cref{sec:conclusion}.

Finally, we also note that under the Planted Clique Hypothesis, we can rule out constant factor approximation but only against polynomial-time algorithms:

\begin{theorem} \label{thm:pc}
Assuming the Planted Clique Hypothesis, there is no polynomial-time algorithm for approximating VC Dimension or Littlestone's Dimension to within any constant factor.
\end{theorem}

We remark that our results are also strong enough to rule out any \emph{fixed-parameter tractable algorithms}\footnote{Please refer to e.g.~\cite{DowneyF13,CyganFKLMPPS15} for background on parameterized complexity.} with non-trivial approximation ratios. We defer the discussion on this to~\Cref{app:fpt-inapprox}.

\subsubsection{Reduction from (Unbalanced) Biclique}
\label{sec:prelim}

All of our results are proved via a simple reduction from an unbalanced version of the Maximum Balanced Complete Bipartite Graph (Maximum Biclique) problem. Indeed, the fact that we can achieve multiple flexible results can be attributed to the rich literature on the Maximum Biclique problem. Furthermore, by identifying this as a hard problem underpinning the approximation of VC Dimension and Littlestone's dimension, our reduction helps demystify the rather specific reductions in \cite{ManurangsiR17}, which seem to be tailored towards starting from the Label Cover problem--and do not seem to be applicable e.g. with Planted Clique hypotheses.

To define the Maximum Biclique problem, we need a few additional terminologies: we use $K_{a, b}$ where $a, b \in \N$ to denote the bipartite complete graph (biclique) with $a$ vertices on one side and $b$ vertices on the other. We say that a bipartite graph $G = (A, B, E)$ \emph{contains $K_{a, b}$} if there exist $a$-size $S \subseteq A$ and $b$-size $T \subseteq B$ such that $S, T$ induces a complete bipartite subgraph. Otherwise, if no such $S, T$ exist, we say that $G$ is \emph{$K_{a, b}$-free}.

Most work in the literature has considered the case of \emph{balanced} biclique, i.e. $a = b$. The Maximum Balanced Biclique problem then asks for the maximum $a$ such that the input graph $G$ contains $K_{a, a}$. For the purpose of stating its inapproximability, it is useful to define the following ``gap'' version\footnote{We provide an additional discussion on gap problems in \Cref{subsec:gap-problems}.}, as stated below. Note that if $(q_1, q_2)$-Gap Biclique Problem cannot be efficiently solved, then there is no efficient $(q_2/q_1)$-approximation algorithm for Maximum Balanced Biclique.

\begin{definition}[$(q_1, q_2)$-Gap Biclique Problem] \label{prob:gap-biclique}
Given a bipartite graph $G = (A, B, E)$, distinguish between the following two cases:
\begin{itemize}
\item (YES) $G$ contains $K_{q_1, q_1}$ as a subgraph.
\item (NO) $G$ is $K_{q_2, q_2}$-free.
\end{itemize}
\end{definition}

The Maximum Balanced Biclique problem turns out to be a challenging research question from hardness of approximation point of view. Despite the well-understood status of its non-bipartite counterpart~\cite{Hastad96,KhotP06,Zuckerman07}, it is not even known if Maximum Balanced Biclique problem is NP-hard to approximate to within a factor of 1.0001. Fortunately for us, several hardness results are known under stronger assumptions~\cite{FK04,BhangaleGHKK16,Khot06,Manurangsi17,Manurangsi17-icalp,ChalermsookCKLM20,ManurangsiRS21}.

As alluded to earlier, the exact problem we will use is not the balanced version above. We need an unbalanced version where one side of the biclique has $t$ vertices and the other side has $2^t$ vertices, as defined below.

\begin{definition}[$(t_1, t_2)$-Gap Exponential Biclique Problem] \label{prob:gap-exp-biclique}
Given a bipartite graph $G = (A, B, E)$, distinguish between the following two cases:
\begin{itemize}
\item (YES) $G$ contains $K_{t_1, 2^{t_1}}$ as a subgraph.
\item (NO) $G$ is $K_{t_2, 2^{t_2}}$-free.
\end{itemize}
\end{definition}

It should be noted that Gap Exponential Biclique is quite different from Gap Biclique. For example, to get a constant gap--say $t_1 \geq 2t_2$--in the former, we need the number of vertices on one side to be $2^{t_1}$ in the YES case versus $2^{t_1/2} = \sqrt{2^{t_1}}$ in the NO case. This seems closer to the setting $q_2 = \sqrt{q_1}$ (i.e. $\sqrt{q_1}$ gap) in the Gap Biclique problem. Furthermore, Gap Exponential Biclique must have $t_1, t_2 \leq \log n$, meaning that the problem is solvable in time $n^{O(\log n)}$. Indeed, in our hardness results below, we have to resort to literature from \emph{parameterized} hardness of Gap Biclique, which also considers the case $q_1, q_2 \ll n$. In some cases, we have to adapt the reductions slightly to get the parameters sufficiently strong for Gap Exponential Biclique.

Let us now turn our attention back to the main reduction, which is from Gap Exponential Biclique to the question of approximating VC / Littlestone's Dimension. In fact, our reduction proves the hardness of approximating \emph{both} VC and Littlestone's Dimension simultaneously. Specifically, we show that the following gap problem is hard\footnote{Recall that the Littlestone's Dimension of any class is no less than its VC Dimension.}:

\begin{definition}[$(d_1, d_2)$-$\VC$-$\Ldim$ Problem] \label{prob:gap-vcldim}
Given a concenpt class $\C$ as a binary matrix, distinguish between the following two cases:
\begin{itemize}
\item (YES) The VC Dimension of $\C$ is at least $d_1$.
\item (NO) The Littlestone's Dimension of $\C$ is less than $d_2$.
\end{itemize}
\end{definition}

With all the components defined, we can now state the properties of our reduction:

\begin{theorem} \label{thm:main-red}
If there is an $T(n)$-time algorithm for the $(d_1, d_2)$-$\VC$-$\Ldim$ Problem, there is an $O(T(n))$-time (randomized) algorithm for the $(d_1/2, 2d_2)$-Gap Exponential Biclique Problem.
\end{theorem}

That is, any hardness of Gap Exponential Biclique with gap $\alpha$ can be translated to hardness of $\VC$-$\Ldim$ Problem with gap $\alpha/4$. This allows us to easily translate the hardness approximating of biclique problems to that of approximating VC / Littlestone's Dimension.

\subsection*{Organization}
The rest of this paper is organized as follows. We provide some additional definitions in \Cref{sec:prelim}. We then present the main reduction (\Cref{thm:main-red}) in \Cref{sec:red}. The subsequent three sections then establish concrete hardness results (stated in \Cref{sec:our-contrib}). Finally, we conclude with some discussion and open questions in \Cref{sec:conclusion}.

\section{Preliminaries}
\label{sec:prelim}

All approximation problems considered in this work are maximization problem. For these problems, we say that an algorithm\footnote{We assume for simplicity that these algorithms are deterministic, although it is trivial to extend the formalism and proofs to randomized approximation algorithms.} has an approximation ratio (or factor) $r$ if its output always lie between $\OPT/r$ and $\OPT$, where $\OPT$ denote the true optimum for the input instance.

\subsection{VC Dimension and Littlestone's Dimension}

For a domain $X$, a concept $c$ is simply a function from $X$ to $\{0, 1\}$. A concept class $\C$ is a set of concepts. For a subset $S \subseteq X$ and a concept $c$, let $c|_S$ denote the restriction of $c$ on $S$, i.e. $c|_S: S \to \{0, 1\}$ such that $c|_S(x) = c(x)$ for all $x \in S$. Furthermore, define $\C|_S := \{c|_S \mid c \in \C\}$.

We now recall the definition of VC Dimension and Littlestone's Dimenion.

\begin{definition}[VC Dimension~\cite{VC}]
The VC dimension of a concept class $\C$, denoted by $\VC(\C)$, is defined as the size of the largest subset $S \subseteq X$ such that $\C|_S$ contains all boolean functions on $S$.
\end{definition}

\begin{definition}[Online Algorithm \& Mistake Bound] \label{def:mistake-bound}
An online learning algorithm $\A$ for a concept class $\C$ is an algorithm that, at each step $i$, is given a sample $x_i$, the algorithm has to make a prediction $z_i$ and afterwards the algorithm is told the correct label $y_i$. The mistake bound of $\A$ for a sequence\footnote{Throughout this work, we assume w.l.o.g. that $x_1, \dots, x_T$ are distinct.} $(x_1, y_1), \dots, (x_T, y_T)$ is defined as the number of incorrect predictions (i.e. $\sum_{i \in [T]} \ind[z_i \ne y_i]$).

The mistake bound of $\A$ for a concept class $\C$ is defined as the maximum mistake bound of $\A$ across all sequences $(x_1, y_1), \dots, (x_T, y_T)$ that are realizable by $\C$ (i.e. there exists $c \in \C$ such that $y_i = c(x_i)$ for all $i \in [T]$).
\end{definition}

\begin{definition}[Littlestone's Dimension~\cite{Littlestone87}]
The Littlestone's Dimension of $\C$, denoted by $\Ldim(\C)$, is defined as the minimum mistake bound for $\C$ across all online algorithms $\A$.
\end{definition}

We note that the Littlestone's Dimension is sometimes defined in terms of mistake trees, but it was shown in~\cite{Littlestone87} that this is equivalent to the above mistake bound definition. We choose to work with the mistake bound definition because it is more convenient in our proofs. 

For convenience, we also extend these biclique-related terminologies to concept classes $\C \subseteq 2^{X}$, where the graph $G$ is define as $(X, \C, E)$ where $(x, c) \in E$ iff $c(x) = 1$.

\subsection{Gap Problems}
\label{subsec:gap-problems}

Gap problems we consider in this work (e.g. \Cref{prob:gap-biclique,prob:gap-exp-biclique,prob:gap-vcldim}) belong to the class of \emph{promise problems}. (See e.g.~\cite{Goldreich06a} for a more thorough treatment on the topic.) A promise problem consists of two disjoint languages $L_{YES}$ and $L_{NO}$. A randomized algorithm $\A$ is said to solve the promise problem $(L_{YES}, L_{NO})$ iff
\begin{itemize}
\item for all inputs $x \in L_{YES}$, $\Pr[\A(x) = YES] \geq 2/3$, and,
\item for all inputs $x \in L_{NO}$, $\Pr[\A(x) = NO] \geq 2/3$.
\end{itemize}
Note that the algorithm are allowed to output arbitrarily for input $x \notin L_{YES} \cup L_{NO}$.

We say that an algorithm is a randomized reduction from a promise problem $(L_{YES}, L_{NO})$ to $(L'_{YES}, L'_{NO})$ if it takes in $x$ and outputs $x'$ such that
\begin{itemize}
\item If $x \in L_{YES}$, then $\Pr[x' \in L'_{YES}] \geq 2/3$.
\item If $x \in L_{NO}$, then $\Pr[x' \in L'_{NO}] \geq 2/3$.
\end{itemize}

It is simple to see that if the reduction runs in $F(n)$ time (which also implies that it products $x'$ of size $N \leq F(n)$) and there is a randomized algorithm for solving $(L'_{YES}, L'_{NO})$ in $T(N)$ time, then there is a randomized algorithm for solving $(L_{YES}, L_{NO})$ in $O(T(F(n)))$ time.

Since most of the biclique hardness results we use are through randomized reductions, we will henceforth drop the ``randomized'' prefix and assume that any discussions related to promise problems are for randomized algorithms and reductions.

\subsection{Exponential Time Hypotheses}

Given a 3CNF formula $\Phi$, we let $\SAT(\Phi)$ denote the fraction of clauses that can be simultaneously satisfied by an assignment. A degree of variable with respect to a formula $\Phi$ is defined as the number of clauses it appears in. Furthermore, we define the $(1, 1 - \mu)$-Gap 3SAT problem to be the problem of distinguishing between (YES) $\SAT(\Phi) = 1$, and (NO) $\SAT(\Phi) < 1 - \mu$.

ETH and Gap-ETH can be formulated as follows.

\begin{assumption}[Exponential Time Hypothesis (ETH)~\cite{ImpagliazzoP01,ImpagliazzoPZ01}]
No $2^{o(N)}$-time algorithm can decide whether any given $N$-variable 3CNF formula $\Phi$ is satisfiable (i.e. $\SAT(\Phi) = 1$). 
\end{assumption}

\begin{assumption}[Gap-Exponential Time Hypothesis (Gap-ETH)~\cite{Dinur16,ManurangsiR17-gap-eth}]
For some constant $d, \mu > 0$, no $2^{o(N)}$-time algorithm can solve the $(1, 1 - \mu)$-Gap 3SAT problem on $N$-variable 3CNF formulae with maximum degree at most $d$.
\end{assumption}

The nearly-linear size PCP of~\cite{Dinur07} gives a lower bound for Gap 3SAT problem under ETH that is slightly weaker than that of Gap-ETH:

\begin{theorem}[Nearly-Linear Size PCP~\cite{Dinur07}] \label{thm:pcp}
Assuming ETH, there exists constant $\upsilon, d, \mu > 0$ such that no $2^{o(N/(\log N)^{\upsilon})}$-time algorithm can solve the $(1, 1 - \mu)$-Gap 3SAT problem on $N$-variable 3CNF formulae with maximum degree at most $d$.
\end{theorem}


\subsection{Planted Clique Hypotheses}

Let $\cG(N, p)$ denote the Erdos-Renyi random graph distribution where there are $N$ vertices and an edge exists between each pair of vertices with probability $p$. Furthermore, let $\cG(N, p, \kappa)$ denote the distribution of a graph sampled from $\cG(N, p)$ but afterwards get planted with a Clique of size $\kappa$.

\begin{assumption}[Planted Clique Hypothesis (e.g.~\cite{Karp76,Jerrum92})]
There exists $\delta > 0$ such that no $N^{O(1)}$-time algorithm $\A$ satisfies both of the following:
\begin{itemize}
\item (Completeness) $\Pr_{G \sim \cG(N, p, \lceil n^\delta\rceil)}[\A(G)=1] \geq 2/3$.
\item (Soundness) $\Pr_{G \sim \cG(N, p)}[\A(G)=1] \leq 1/3$.
\end{itemize}
\end{assumption}

\begin{assumption}[Strongish Planted Clique Hypothesis~\cite{ManurangsiRS21}]
There exists $\delta > 0$ such that no $N^{o(\log N)}$-time algorithm $\A$ satisfies both of the following:
\begin{itemize}
\item (Completeness) $\Pr_{G \sim \cG(N, p, \lceil n^\delta\rceil)}[\A(G)=1] \geq 2/3$.
\item (Soundness) $\Pr_{G \sim \cG(N, p)}[\A(G)=1] \leq 1/3$.
\end{itemize}
\end{assumption}

\section{Reducing Biclique to VC/Littlestone's Dimension}
\label{sec:red}

In this section, we present our main reduction (\Cref{thm:main-red}). Before we do so, let us describe the high-level overview of the reduction. Let $G = (A, B, E)$ be an input instance of $(d_1, d_2)$-Gap Exponential Biclique. Perhaps the simplest reduction one can attempt is to simply let the class $\C$ be corresponding to the adjacency matrix of $G$, i.e. $X = A, \C = B$ and $c(a) = 1$ iff $(a, c) \in E$. 

This simple reduction actually ``almost works''. Specifically, the soundness (i.e. NO case) actually holds. This can easily be seen for VC Dimension, because if $\C|_S$ consists of all boolean functions on $S$, then it must contain a $K_{\lfloor  |S|/2 \rfloor, 2^{\lfloor  |S|/2 \rfloor}}$. This means that in the NO case, the VC Dimension must be less than $2d_2$. Below, in \Cref{lem:littlestone-from-biclique}, we show that the same bound holds even against Littlestone's Dimension.

While the soundness holds, the completeness (i.e. YES case) of the reduction clearly fails: even if $G$ is a complete bipartite graph, the corresponding $\C$ has VC Dimension of just zero. The problem here lies in the fact that ``there are too many ones'' in the adjacency matrix. It turns out that this is very simple to fix: just flip each entry in the adjacency matrix to zero w.p. 1/2. 
This means that the submatrix induced by the biclique is now a random matrix; it is easy to show that such a matrix has high VC Dimension with large probability (\Cref{lem:vc-random}).
This completes the high-level overview of the reduction.


\subsection{Some Helpful Lemmas}
\label{subsec:helpful-lemmas}

Before we formalize the reduction, let us state a few necessary lemmas. First is the aforementioned bound for the Littlestone's Dimension for biclique-free concept class, which will be used in the soundness proof.

\begin{lemma} \label{lem:littlestone-from-biclique}
Let $t \in \N$ and $\C$ be any $K_{t, 2^t}$-free concept class. Then, we must have $\Ldim(\C) < 2t$. 
\end{lemma}

To prove the above lemma, it is convenient to define $\C[{(x_1, y_1), \dots, (x_T, y_T)}]$ for any sequence $(x_1, y_1), \dots, (x_T, y_T)$ as the class of concepts in $\C$ that is consistent with the sequence. More formally, $$\C[{(x_1, y_1), \dots, (x_T, y_T)}] := \{c \in \C \mid \forall i \in [T], c(x_i) = y_i\}.$$

\begin{proof}[Proof of \Cref{lem:littlestone-from-biclique}]
To give the desired upper bound on $\Ldim(\C)$, it suffices to give an online algorithm $\A$ that achives a mistake bound of less than $2t$ on $\C$. After receiving $(x_1, y_1), \dots, (x_{i - 1}, y_{i - 1})$ and $x_i$, the algorithm gives prediction $z_i$ based on the following rule:
\begin{enumerate}
\item If $|\C[{(x_1, y_1), \dots, (x_{i - 1}, y_{i - 1})}]| \geq 2^t$, the algorithm outputs 0.
\item Otherwise, the algorithm outputs $\argmax_{z \in \{0, 1\}} |\C[{(x_1, y_1), \dots, (x_{i - 1}, y_{i - 1}), (x_i, z)}]|$ (tie broken arbritrarily).
\end{enumerate}
Note that since $\C[{(x_1, y_1), \dots, (x_{i - 1}, y_{i - 1})}]$ is non-increasing, the algorithm will first be in the first stage (in the first case) and then moves to the second stage (the second case).

Since $\C$ is $K_{t, 2^t}$-free, $\A$ will make at most $t$ mistakes in the first stage\footnote{Otherwise, $x_1, \dots, x_{i-1}$ and the points in $\C[{(x_1, y_1), \dots, (x_{i - 1}, y_{i - 1})}]$ would induce $K_{t, 2^t}$.}. Once $\A$ reaches the second stage, the size $|\C[{(x_1, y_1), \dots, (x_{i - 1}, y_{i - 1})}]|$ reduces by a factor of (at least) two with each mistake. Since it starts off with value less than $2^t$, $\A$ must make less than $t$ mistakes in the second stage. Therefore, in total $\A$ will make less than $2t$ mistakes.
\end{proof}





The second lemma we need, which will be used in the completeness argument, is that a class consisting of independent random concepts have large VC Dimensions:

\begin{lemma} \label{lem:vc-random}
Let $S$ be any $t$-size set and $\C \subseteq \{0, 1\}^S$ be a concept class of size $2^t$, where each concept is independently uniformly sampled at random from $\{0, 1\}^S$. If $t \geq 4$, then $\VC(\C) \geq t/2$ with probability at least 2/3.
\end{lemma}

\begin{proof}
Consider any subset $S' \subseteq S$ of size $t' := \lceil t/2 \rceil$. We have
\begin{align*}
\Pr[\VC(\C) < t/2] 
&\leq \Pr[\exists f \in \{0,1\}^{S'}, f \notin \C|_{S'}] \\
&\leq \sum_{f \in \{0, 1\}^{S'}} \Pr[f \notin \C|_{S'}] \\
&= \sum_{f \in \{0, 1\}^{S'}} (1 - 1/2^{t'})^{|\C|} \\
&\leq 2^{t'} \cdot e^{-|\C|/2^{t'}} = 2^{\lceil t/2 \rceil} \cdot e^{-2^{\lfloor t/2 \rfloor}} < 1/3,
\end{align*}
where the second inequality is due to the union bound and the first equality is due to the fact that each function is $\C$ is picked independently u.a.r. from $\{0, 1\}^S$.
\end{proof}

\subsection{The Reduction}

We are now ready to prove \Cref{thm:main-red}. 

\begin{proof}[Proof of \Cref{thm:main-red}]
We will provide a reduction from the $(d_1, d_2)$-Gap Exponential Biclique problem to the $(d_1/2, 2d_2)$-VC-Ldim problem such that the problem size remains the same, which implies the theorem statement. (See discussion in \Cref{subsec:gap-problems}.)

Given an instance $G = (A, B, E)$ for the $(d_1, d_2)$-Gap Exponential Biclique Problem.
Let $X = A$ and, for every $b \in B$, create a concept $c_b$ in $\C$ where $c_b(a)$ for each $a \in A$ is defined as follows.
\begin{itemize}
\item If $(a, b) \notin E$, let $c_b(a) = 0$.
\item If $(a, b) \in E$, let $c_b(a) \in \{0, 1\}$ be drawn uniformly at random.
\end{itemize}

\noindent \textbf{(Completeness)}
We may assume w.l.o.g. that $d_1 \geq 4$; otherwise, the problem can be trivially solved in polynomial time. 
Suppose that $G$ contains $K_{d_1, 2^{d_1}}$, i.e. there exists $S \subseteq A, T \subseteq B$ of sizes $d_1, 2^{d_1}$ respectively such that $(s, t) \in E$ for all $s \in S, t \in T$. This, together with the definition of $\C$, means that $c_t|_S$ for each $t \in T$ is drawn independently uniformly at random among $\{0, 1\}^S$. Thus, we may apply \Cref{lem:vc-random} to conclude that $\VC(\C) \geq \VC(\C|_S)$ is at least $d_1/2$ with probability 2/3.

\noindent \textbf{(Soundness)} If $G$ is $K_{d_2, 2^{d_2}}$-free, then $\C$ is also $K_{d_2, 2^{d_2}}$-free (because the corresponding graph of $\C$ is a subgraph of $G$). Therefore, \Cref{lem:littlestone-from-biclique} ensures that $\Ldim(\C_i) \leq 2d_2$.
\end{proof}

\section{(Gap-)ETH-Hardness of Exponential Biclique}

Having presented the reduction, we will next prove concrete hardness results, starting with those based on Gap-ETH and ETH. To prove these, we start by making a straightforward observation: there is the trivial reduction from $(d_1, d_2)$-Gap Biclique to $(\lfloor \log d_1 \rfloor, d_2)$-Gap Exponential Biclique (i.e. keeping the input graph exactly the same). This yields the following result.

\begin{observation} \label{obs:balanced-to-unbalanced}
If there exists an $T(n)$-time algorithm for the $(\lfloor \log d_1 \rfloor, d_2)$-Gap Exponential Biclique Problem, there exists an $T(n)$-time algorithm for the $(d_1, d_2)$-Gap Biclique Problem.
\end{observation}

We will also need the following reduction from Gap 3SAT to Gap Biclique due to \cite{ChalermsookCKLM20} (which is in turn a modification of that from~\cite{Manurangsi17}).

\begin{theorem}[{\cite[Theorem 5.13]{ChalermsookCKLM20}}] \label{thm:red-sat-biclique}
For any constants $d, \mu > 0$, there exists a constant $\gamma > 0$ and a reduction from $(1, 1 - \mu)$-Gap 3SAT problem on $n$-variable 3CNF formulae with maximum degree at most $d$ to $(n^{\gamma / \sqrt{r}}, r)$-Gap Biclique on graphs of size $n$ for any sufficiently large $r$. Furthermore, $n = 2^{\Theta_{d,\mu}(N / \sqrt{r})}$ and the reduction runs in $n^{O(1)}$ time.
\end{theorem}

\subsection{Gap-ETH-Hardness}

We will now prove the Gap-ETH hardness results (\Cref{thm:gap-eth-inapprox,thm:gap-eth-time}) by plugging in an appropriate value of $r$ in \Cref{thm:red-sat-biclique} and then apply our reduction (\Cref{thm:main-red}) afterwards. For the $o(\log n)$ factor inapproximability result, we set $r$ to be slowly growing function of $N$, as formalized below.

\begin{proof}[Proof of \Cref{thm:gap-eth-inapprox}]
Let $\alpha(n)$ denote any function\footnote{Henceforth, we assume w.l.o.g. that the approximation ratio (denoted by $\alpha(n)$) is a non-decreasing function of the problem size $n$; this is w.l.o.g. as otherwise we can always consider $\tilde{\alpha}(n) := \max_{n' \leq n} \alpha(n)$ instead.} such that $\alpha(n) = o(\log n)$ (i.e. $\lim_{n \to \infty} \frac{\alpha(n)}{\log n} = 0$).  We will prove below that there is no polynomial-time algorithm for $(4 \alpha(n) d, d)$-Gap Exponential Biclique, where $d$ is a parameter to be set below. By \Cref{thm:main-red}, this immediately implies that there is no polynomial-time $\alpha(n)$-approximation algorithm for VC dimension or Littlestone's dimension.

Let $\Phi$ be the input to the $(1, 1 - \mu)$-Gap 3SAT problem. Let $r = \lceil \sigma \cdot \sqrt{N / \alpha(2^N)} \rceil$ where $\sigma > 0$ is a sufficiently small constant to be chosen later. Note here that $r = \omega(1)$ because $\alpha(2^N) = o(\log(2^N)) = o(N)$. We apply the reduction in \Cref{thm:red-sat-biclique} to produce an input graph $G'$ of size $n = 2^{\Theta(N/\sqrt{r})} \leq 2^{o(n)}$ for $(n^{\gamma / \sqrt{r}}, r)$-Gap Biclique. \Cref{obs:balanced-to-unbalanced} also implies that this is an instance for $(\lfloor\log(n^{\gamma / \sqrt{r}})\rfloor, r)$-Gap Exponential Biclique. Let $t_1 = \lfloor\log(n^{\gamma / \sqrt{r}})\rfloor = \Theta(N/r)$ and $d = t_2 = r$. From our choice of $r$, we have $t_1 / t_2 = \Theta(1/\sigma^2) \cdot \alpha(2^N)$. Thus, by taking $\sigma$ to be a sufficiently small constant, the ratio between $t_1, t_2$ is at least $4\alpha(2^N)$, which is in turn no less than $4\alpha(n)$ for any sufficiently large $N$. 

Thus, if there were a polynomial-time algorithm for $(4\alpha(n)d, d)$-Gap Exponential Biclique Problem, then we could apply the above reduction and run it on the resulting graph $G'$ to solve the $(1, 1 - \mu)$-Gap 3SAT problem in time $2^{o(N)}$. This would violate Gap-ETH.
\end{proof}

For \Cref{thm:gap-eth-time}, we instead choose $r = \Theta(\sqrt{N})$.

\begin{proof}[Proof of \Cref{thm:gap-eth-time}]
Let $C > 1$ be any constant. We will prove below that there is no $n^{o\left({(\log n)^{1/3}}\right)}$-time algorithm for $(4 C d, d)$-Gap Exponential Biclique, where $d$ is a parameter to be set below. By \Cref{thm:main-red}, this immediately implies that there is no $n^{o\left({(\log n)^{1/3}}\right)}$-time $C$-approximation algorithm for VC dimension or Littlestone's dimension.

Let $\Phi$ be the input to the $(1, 1 - \mu)$-Gap 3SAT problem. Let $r = \lceil \sigma \cdot \sqrt{N} \rceil$ where $\sigma > 0$ is a sufficiently small constant to be chosen later. We apply the reduction in \Cref{thm:red-sat-biclique} to produce an input graph $G'$ of size $n = 2^{\Theta(N/\sqrt{r})} = 2^{\Theta(N^{3/4})}$ for $(n^{\gamma / \sqrt{r}}, r)$-Gap Biclique. \Cref{obs:balanced-to-unbalanced} also implies that this is an instance for $(\lfloor\log(n^{\gamma / \sqrt{r}})\rfloor, r)$-Gap Exponential Biclique. Let $t_1 = \lfloor\log(n^{\gamma / \sqrt{r}})\rfloor = \Theta(N/r)$ and $d = t_2 = r$. From our choice of $r$, we have $t_1 / t_2 = \Theta(1/\sigma^2)$. Thus, by taking $\sigma$ to be a sufficiently small constant, the ratio between $t_1, t_2$ is at least $4C$ as desired. 

Thus, if there were a $n^{o((\log n)^{1/3})}$-time algorithm for $(4Cd, d)$-Gap Exponential Biclique Problem, then we could apply the above reduction and run it on the resulting graph $G'$ to solve the $(1, 1 - \mu)$-Gap 3SAT problem in time $(2^{\Theta(N^{3/4})})^{o(\log(2^{\Theta(N^{3/4})})^{1/3})} = 2^{o(N)}$. This would violate Gap-ETH.
\end{proof}

\subsection{ETH-Hardness}

The proofs of \Cref{thm:eth-inapprox,thm:eth-time} are nearly identical to those of \Cref{thm:gap-eth-inapprox,thm:gap-eth-time}, except that we have to set $r$ to be larger than the corresponding Gap-ETH-based results by a polylogarithmic in $N$ factor. This is to compensate for the weaker running time lower bound we have under ETH from \Cref{thm:pcp}.

\begin{proof}[Proof of \Cref{thm:eth-inapprox}]
Let $\xi = 4\upsilon$ where $\upsilon$ is the constant from \Cref{thm:pcp}.
Let $\alpha(n)$ denote any function such that $\alpha(n) = o\left(\frac{\log n}{(\log \log n)^\xi}\right)$.  We will prove below that there is no polynomial-time algorithm for $(4 \alpha(n) d, d)$-Gap Exponential Biclique, where $d$ is a parameter to be set below. By \Cref{thm:main-red}, this immediately implies that there is no polynomial-time $\alpha(n)$-approximation algorithm for VC dimension or Littlestone's dimension.

Let $\Phi$ be the input to the $(1, 1 - \mu)$-Gap 3SAT problem. Let $r = \lceil \sigma \cdot \sqrt{N / \alpha(2^N)} \rceil$ where $\sigma > 0$ is a sufficiently small constant to be chosen later. Note here that $r = \omega((\log N)^{2\upsilon})$ because $\alpha(2^N) = o\left(\frac{\log(2^N)}{(\log \log(2^N))^\xi}\right) = o\left(\frac{N}{(\log N)^{\xi}}\right)$. We apply the reduction in \Cref{thm:red-sat-biclique} to produce an input graph $G'$ of size $n = 2^{\Theta(N/\sqrt{r})} \leq 2^{o(N / (\log N)^{\upsilon})}$ for $(n^{\gamma / \sqrt{r}}, r)$-Gap Biclique. \Cref{obs:balanced-to-unbalanced} also implies that this is an instance for $(\lfloor\log(n^{\gamma / \sqrt{r}})\rfloor, r)$-Gap Exponential Biclique. Let $t_1 = \lfloor\log(n^{\gamma / \sqrt{r}})\rfloor = \Theta(N/r)$ and $d = t_2 = r$. From our choice of $r$, we have $t_1 / t_2 = \Theta(1/\sigma^2) \cdot \alpha(2^N)$. Thus, by taking $\sigma$ to be a sufficiently small constant, the ratio between $t_1, t_2$ is at least $4\alpha(2^N)$, which is in turn no less than $4\alpha(n)$ for any sufficiently large $N$. 

Thus, if there were a polynomial-time algorithm for $(4\alpha(n)d, d)$-Gap Exponential Biclique Problem, then we could apply the above reduction and run it on the resulting graph $G'$ to solve the $(1, 1 - \mu)$-Gap 3SAT problem in time $2^{o(N / (\log N)^{\upsilon})}$. From \Cref{thm:pcp}, this would violate ETH.
\end{proof}

\begin{proof}[Proof of \Cref{thm:eth-time}]
Let $\xi = \upsilon$ be the constant from \Cref{thm:pcp}.
Let $C > 1$ be any constant. We will prove below that there is no $n^{o\left(\frac{(\log n)^{1/3}}{(\log \log n)^{\xi}}\right)}$-time algorithm for $(4 C d, d)$-Gap Exponential Biclique, where $d$ is a parameter to be set below. By \Cref{thm:main-red}, this immediately implies that there is no $n^{o\left(\frac{(\log n)^{1/3}}{(\log \log n)^{\xi}}\right)}$-time $C$-approximation algorithm for VC dimension or Littlestone's dimension.

Let $\Phi$ be the input to the $(1, 1 - \mu)$-Gap 3SAT problem. Let $r = \lceil \sigma \cdot \sqrt{N} \rceil$ where $\sigma > 0$ is a sufficiently small constant to be chosen later. We apply the reduction in \Cref{thm:red-sat-biclique} to produce an input graph $G'$ of size $n = 2^{\Theta(N/\sqrt{r})} = 2^{\Theta(N^{3/4})}$ for $(n^{\gamma / \sqrt{r}}, r)$-Gap Biclique. \Cref{obs:balanced-to-unbalanced} also implies that this is an instance for $(\lfloor\log(n^{\gamma / \sqrt{r}})\rfloor, r)$-Gap Exponential Biclique. Let $t_1 = \lfloor\log(n^{\gamma / \sqrt{r}})\rfloor = \Theta(N/r)$ and $d = t_2 = r$. From our choice of $r$, we have $t_1 / t_2 = \Theta(1/\sigma^2)$. Thus, by taking $\sigma$ to be a sufficiently small constant, the ratio between $t_1, t_2$ is at least $4C$ as desired. 

Thus, if there were a $n^{o\left(\frac{(\log n)^{1/3}}{(\log \log n)^{\xi}}\right)}$-time algorithm for $(4Cd, d)$-Gap Exponential Biclique Problem, then we could apply the above reduction and run it on the resulting graph $G'$ to solve the $(1, 1 - \mu)$-Gap 3SAT problem in time $(2^{\Theta(N^{3/4})})^{o\left(\frac{\log(2^{\Theta(N^{3/4})})^{1/3}}{(\log \log(2^{\Theta(N^{3/4})})^{\xi}}\right)} = 2^{o(N / (\log N)^{\upsilon})}$. From \Cref{thm:pcp}, this would violate ETH.
\end{proof}

\section{Hardness from Planted Clique Hypothesis}

We next move on to prove hardness result based on the Planted Clique Hypothesis (\Cref{thm:pc}). We start by recalling that the Planted Clique Hypothesis by itself already implies a fairly strong hardness of Gap Biclique:

\begin{lemma}[Folklore] \label{lem:pc-biclique}
Assuming the Planted Clique Hypothesis, for some constants $\delta \in (0, 1), \zeta > 1$, there is no polynomial-time algorithm for $(\lceil N^{\delta} \rceil, \zeta \lceil \log N\rceil)$-Gap Biclique.
\end{lemma}

The above result is folklore; for completeness, we provide a proof sketch in \Cref{app:pc-biclique}.

While \Cref{lem:pc-biclique} is indeed a strong result, it is not yet enough for us. Specifically, if we want to apply \Cref{obs:balanced-to-unbalanced} from the previous section, then we would now only reduce to the $(\delta \log N, \zeta \log N)$-Gap Exponential Biclique problem. However, this is trivial because $\delta < 1 < \zeta$.

\subsection{One-Sided Graph Product}

The above problem is due to the fact that the YES case is too small to apply the observation directly. To boost the YES case, we use the following ``one-sided'' graph product, which helps boost one side of the biclique in the YES case while not increasing the NO case. The reduction and its main properties are described below.

\begin{figure}[h!]
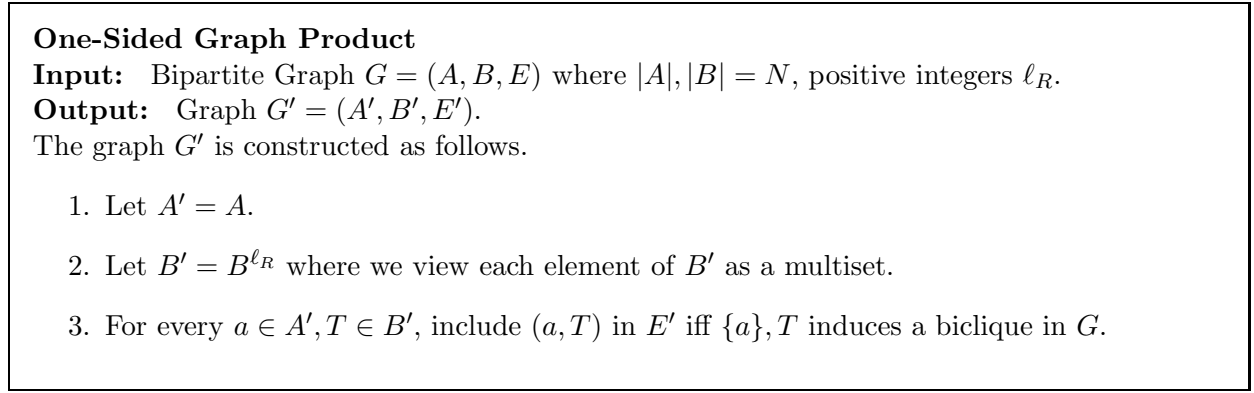

\begin{framed}
\textbf{One-Sided Graph Product} \\
\textbf{Input: } Bipartite Graph $G = (A, B, E)$ where $|A|, |B| = N$, positive integers $\ell_R$. \\
\textbf{Output: } Graph $G' = (A', B', E')$. \\
The graph $G'$ is constructed as follows.
\begin{enumerate}
\item Let $A' = A$.
\item Let $B' = B^{\ell_R}$ where we view each element of $B'$ as a multiset.
\item For every $a \in A', T \in B'$, include $(a, T)$ in $E'$ iff $\{a\}, T$ induces a biclique in $G$.
\end{enumerate}
\end{framed}
\caption{One-Sided Graph Product. \label{fig:graph-prod-one-sided}}
\end{figure}

\begin{lemma} \label{lem:graph-prod-one-sided}
Let $\delta \in (0, 1)$ be a constant and $N$ be sufficiently large (depending on $\delta$). 
Let $G, G'$ be as in \Cref{fig:graph-prod-one-sided}. Then, we have
\begin{itemize}
\item (Completeness) If $G$ contains $K_{q_1, q_1}$, $q_1 \geq t_1$ and $\ell_R \log(q_1) \geq t_1$, then $G'$ contains $K_{t_1, 2^{t_1}}$.
\item (Soundness) If $G$ is $K_{q_2, q_2}$-free, $q_2 \leq t_2$ and $\ell_R \log(q_2) \leq t_2$, then $G'$ is $K_{t_2, 2^{t_2}}$-free.
\end{itemize}
\end{lemma}

\begin{proof}
\textbf{(Completeness)} Let $P \subseteq A, Q \subseteq B$ be a $K_{q_1, q_1}$-biclique in $G$. Observe that $P$ and $Q^{\ell_R}$ induces a biclique in $G'$. We also have $|Q^{\ell_R}| = q_1^{\ell_R} \geq 2^{t_1}$ where the inequality follows from the second assumption on the parameters. Thus, we can conclude that $G'$ contains $K_{t_1, 2^{t_1}}$ as desired.

\textbf{(Soundness)} 
Suppose for the sake of contradiction that $G'$ contains $K_{t_2, 2^{t_2}}$, i.e. there exists $P \subseteq A', \cT \subseteq B'$ of sizes $t_2, 2^{t_2}$ respectively that induces a biclique. By definition, this also means that $P$ and $Q := \left(\bigcup_{T \in \cT} T\right) \subseteq B$ induce a biclique in the graph $G$. Since $\cT \subseteq Q^{\ell_R}$, we must have $Q \geq |\cT|^{1 / \ell_R} \geq q_2$, where the second inequality follows from our assumption that $\ell_R \log(q_2) \leq t_2$. Recall also that we assume that $t_2 \geq q_2$. As a result, we can conclude that $G$ contains $K_{q_2, q_2}$.
\end{proof}

\subsection{Proof of \Cref{thm:pc}}

We can now prove the desired hardness (\Cref{thm:pc}) via the one-sided graph product where we pick $\ell_R$ to be a sufficiently large constant.

\begin{proof}[Proof of \Cref{thm:pc}]
Let $C > 1$ be any constant. We will prove below that there is no polynomial-time algorithm for $(4 C d, d)$-Gap Exponential Biclique, where $d$ is a parameter to be set below. By \Cref{thm:main-red}, this immediately implies that there is no polynomial-time $C$-approximation algorithm for VC dimension or Littlestone's dimension.

Let $G$ be the input to the $(\lceil N^{\delta} \rceil, \zeta \lceil \log N\rceil)$-Gap Biclique problem.
Let $q_1 = \lceil N^{\delta} \rceil, d = q_2 = \lceil \zeta \log N \rceil$.
We use the reduction in \Cref{fig:graph-prod-one-sided} with $t_1 = 4Cd, t_2 = d, \ell_R = \lceil t_1 / \log(q_1) \rceil$. The reduction runs in polynomial time (because $\ell_R \leq O(1)$). We now check the completeness and soundness:
\begin{itemize}
\item (Completeness) By our setting of parameters, $\ell_R\log(q_1) \geq t_1$. Furthermore, for any sufficiently large $N$, we have $q_1 = \Theta(N^{\delta})$ is at least $t_1 = \Theta(\log N)$. Thus, applying the completeness of \Cref{lem:graph-prod-one-sided}, we conclude that, if $G$ contains $K_{q_1, q_1}$, then the $G'$ contains $K_{t_1, 2^{t_1}}$.
\item (Soundness) By our setting of parameters, $q_2 = t_2$. Furthermore, $\ell_R \log(q_2) = O(\log \log N)$ must be less than $t_2 = \Theta(\log N)$ for any sufficiently large $N$. Thus, applying the completeness of \Cref{lem:graph-prod-one-sided}, we can conclude that, if $G$ is $K_{q_2, q_2}$-free, then the $G'$ is $K_{t_1, 2^{t_1}}$-free.
\end{itemize} 

Thus, if there were a polynomial-time algorithm for $(4Cd, d)$-Gap Exponential Biclique Problem, then we could apply the above reduction and run it on the resulting graph $G'$ to solve the  $(\lceil N^{\delta} \rceil, \zeta \lceil \log N\rceil)$-Gap Biclique problem in polynomial time. From \Cref{lem:pc-biclique}, this would violate the Strongish Planted Clique Hypothesis.
\end{proof}

\section{Hardness from Strongish Planted Clique Hypothesis}

Finally, we will prove the hardness results based on the Strongish Planted Clique Hypothesis (\Cref{thm:spc-approx,thm:spc-time}). Again, we start by recalling that the Strongish Planted Clique Hypothesis already implies a fairly strong hardness of Gap Biclique, as stated below. (Proof sketch in \Cref{app:pc-biclique}.)

\begin{lemma}[Folklore] \label{lem:spc-biclique}
Assuming the Strongish Planted Clique Hypothesis, for some constants $\delta \in (0, 1), \zeta > 1$, there is no $N^{o(\log N)}$-time algorithm for $(\lceil N^{\delta} \rceil, \lceil \zeta \cdot \log N\rceil)$-Gap Biclique.
\end{lemma}

The proof of \Cref{thm:spc-time} is exactly the same as \Cref{thm:pc} except that we start with an $N^{\Omega(\log N)}$ running time lower bound (in \Cref{lem:spc-biclique}) instead of just polynomial running time lower bound (in \Cref{lem:pc-biclique}), so we end up with an $n^{\Omega(\log n)}$ running time lower bound as well.

For the remainder of this section, we focus on proving the tight inapproximability ratio (\Cref{thm:spc-approx}).

\subsection{Two-Sided Randomized Graph Product}

Again, we employ graph products, but this time the product is two-sided and furthermore is randomized. This ``two-sided randomized graph product'' has been used several times in literature (e.g.~\cite{Khot06}) but we will consider different parameter regimes compared to previous works. Therefore, we state the randomized graph product and its properties in full below.

\begin{figure}[h!]
\begin{framed}
\textbf{Two-Sided Randomized Graph Product~\cite{BermanS89}} \\
\textbf{Input: } Bipartite Graph $G = (A, B, E)$ where $|A|, |B| = N$, positive integers $n, \ell$. \\
\textbf{Output: } Graph $G' = (A', B', E')$. \\
The graph $G'$ is constructed as follows.
\begin{enumerate}
\item For each $i \in [n]$, independently sample $S_i \sim A^{\ell}$. Then, let $A' = \{S_1, \dots, S_n\}$.
\item For each $j \in [n]$, independently sample $T_i \sim B^{\ell}$. Then, let $B' = \{T_1, \dots, T_n\}$.
\item For every $i, j \in [n]$, include $(S_i, T_j)$ in $E'$ iff $S_i, T_j$ induces a biclique in $G$.
\end{enumerate}
\end{framed}
\caption{Two-Sided Randomized Graph Product}
\label{fig:graph-prod}
\end{figure}

\begin{lemma} \label{lem:graph-prod-main}
Let $\delta \in (0, 1)$ be a constant and $N$ be sufficiently large (depending on $\delta$). Furthermore, suppose that $n \geq 10 2^{t_1} \cdot (N/q_1)^\ell, n \le 1000 N^{(1-0.5\delta)\ell}, 20 \le \ell$ and $0.005\delta \cdot t_2 \cdot \ell \geq q_2$. Then, the reduction in \Cref{fig:graph-prod} is a reduction from $(q_1, q_2)$-Gap Biclique to $(2^{t_1}, t_2)$-Gap Biclique.
\end{lemma}

To prove \Cref{lem:graph-prod-main}, we require a lemma showing that $S_1, \dots, S_n, T_1, \dots, T_n$ (when viewed as sets) are ``dispersers'', meaning that a union of a certain number of them is sufficiently large. \cite{ManurangsiRS21} also used such a property and their lemma (stated below) will be sufficient for us.

\begin{lemma}[{\cite[Lemma 7]{ManurangsiRS21}}] \label{lem:disperser}
Let $\gamma > 0$, and suppose $n \le 1000 N^{(1-\gamma)\ell}$, $20\le \ell$. Let $S_1, \dots, S_n, T_1, \dots, T_n$ be as sampled as in  \Cref{fig:graph-prod}.
Then, with probability at least 0.9, the following event occurs: for every $M \subseteq [n]$ with $|M| \le N^{0.99\gamma} / \ell$, we have $|\bigcup_{i \in M} S_i| \geq 0.01 \gamma |M| \ell$ and $|\bigcup_{i \in M} T_i| \geq 0.01 \gamma |M| \ell$.
\end{lemma}

\begin{proof}[Proof of \Cref{lem:graph-prod-main}]
\textbf{(Completeness)} Let $P \subseteq A, Q \subseteq B$ be a $K_{q_1, q_1}$-biclique in $G$. Observe that $(A' \cap P^\ell)$ and $(B' \cap Q^\ell)$ induce a biclique in $G'$. Notice also that each $S_i$ belongs to $P^\ell$ with probability $(q_1 / N)^{\ell}$. Therefore, we have $\E[|A' \cap P^\ell|] = n \cdot (q_1 / N)^{\ell} \geq 10 2^{t_1}$, where the inequality follows from our assumption on parameters. Applying standard concentration bounds, we can conclude that $\Pr[|A' \cap P^\ell| \geq 2^{t_1}] \geq 0.95$. An analogous argument shows that $\Pr[|B' \cap Q^\ell| \geq 2^{t_1}] \geq 0.95$. Applying the union bound, $G'$ contains $K_{2^{t_1}, 2^{t_1}}$ with probability at least 0.9.

\textbf{(Soundness)} 
From our assumption on the parameters, we can apply \Cref{lem:disperser} (with $\gamma = 0.5\delta$), which guarantee that with probability 0.9, the following holds: 
\begin{align} \label{eq:expansion}
\forall M \in \binom{[n]}{t_2}, & &\left|\bigcup_{i \in M} S_i\right| \geq q_2 &\text{ and } \left|\bigcup_{i \in M} T_i\right| \geq q_2.
\end{align} We will prove that, when this holds, if $G$ is $K_{q_2, q_2}$-free, then $G'$ is $K_{t_2, t_2}$-free. 

Suppose contrapositively that $G'$ contains $K_{t_2, t_2}$, i.e. there exists $\cP \subseteq A', \cQ \subseteq B'$ each of size $t_2$ respectively that induces a biclique in $G'$. By definition, this also means that $P := \left(\bigcup_{S \in \cP} S\right) \subseteq A$ and $Q := \left(\bigcup_{T \in \cQ} T\right) \subseteq B$ induce a biclique in the graph $G$. By~\eqref{eq:expansion}, we must have $|P|, |Q| \geq q_2$, meaning that $G$ contains $K_{q_2, q_2}$.
\end{proof}

\subsection{Proof of \Cref{thm:spc-approx}}

We can now prove \Cref{thm:spc-approx} by applying the above two-sided randomized graph product with appropriate parameters (e.g. $\ell$ is $o(\log N)$).

\begin{proof}[Proof of \Cref{thm:spc-approx}]
Let $\alpha(n)$ denote any function such that $\alpha(n) = o(\log n)$.  We will prove below that there is no polynomial-time algorithm for $(4 \alpha(n) d, d)$-Gap Exponential Biclique, where $d$ is a parameter to be set below. By \Cref{thm:main-red}, this immediately implies that there is no polynomial-time $\alpha(n)$-approximation algorithm for VC dimension or Littlestone's dimension.

Let $G$ denote the input to the $(\lceil N^{\delta} \rceil, \zeta \lceil \log N\rceil)$-Gap Biclique problem.
Let $q_1 = \lceil N^{\delta} \rceil, q_2 = \lceil \zeta \log N \rceil$ and $d = \left\lceil(\log N)^{2/3} / \alpha(N^{\log N})^{1/3}\right\rceil$. Note that $d = \omega(1)$ because $\alpha(N^{\log N}) = o(\log(N^{\log N})) = o(\log^2 N)$. We use the reduction in \Cref{fig:graph-prod} with $t_1 = \alpha(N^{\log N}) \cdot d, t_2 = d, \ell = \max\left\{20, \left\lceil \frac{200 q_2}{t_2 \delta} \right\rceil\right\}, n = \lceil 10 2^{t_1} \cdot (N/q_1)^{\ell} \rceil$. Note that $\ell = O(q_2/t_2) = o(\log N)$. Furthermore, $t_1 / \ell = O(\alpha(N^{\log N}) \cdot d^2 / \log N) = O(\alpha(N^{\log N})^{1/3} \cdot (\log N)^{1/3}) = o(\log N)$. Therefore, we have $n = O(2^{t_1} \cdot N^{(1 - \delta)\ell}) = N^{(1 - \delta + o(1))\ell} \leq N^{o(\log N)}$. The second-to-last inequality also implies that $n \leq N^{(1 - 0.5\delta)\ell}$ for any sufficiently large $N$. In other words, for any sufficiently large $N$, the parameters satisfy conditions in \Cref{lem:graph-prod-main}.


Thus, the reduction runs in $n^{O(1)} = N^{o(\log N)}$ time and produces an instance to the $(2^{t_1}, t_2)$-Gap Biclique Problem. By \Cref{obs:balanced-to-unbalanced}, this is also an instance for the $(t_1, t_2)$-Gap Exponential Biclique problem. By definition, $t_1/t_2 = \alpha(N^{\log N}) \geq \alpha(n)$ for any sufficiently large $N$.

Thus, if there were a polynomial-time algorithm for $(4\alpha(n)d, d)$-Gap Exponential Biclique Problem, then we could apply the above reduction and run it on the resulting graph $G'$ to solve the $(\lceil N^{\delta} \rceil, \zeta \lceil \log N\rceil)$-Gap Biclique problem in time $N^{o(\log N)}$. From \Cref{lem:spc-biclique}, this would violate the Strongish Planted Clique Hypothesis.
\end{proof}

\section{Conclusion and Discussion}
\label{sec:conclusion}

In this work, we establish several hardness of approximation results and running time lower bounds for approximating VC Dimension and Littlestone's Dimension. For polynomial-time algorithms, we rule out $o(\log n)$-approximation~\Cref{thm:eth-inapprox,thm:spc-approx}, which is tight. For any constant factor approximation, we rule out algorithms that runs in $n^{o(\log n)}$ time under the Strongish Planted Clique Hypothesis but only $n^{\tilde{o}(\log^{1/3} N)}$ time under ETH/Gap-ETH. The latter is not a coincidence: the best running time lower bounds for finding $k$ balanced biclique known under ETH/Gap-ETH is only $n^{\Omega(\sqrt{k})}$ even for the exact version of the problem~\cite{Lin15} while the trivial algorithm runs in $n^{O(k)}$ time. Closing this gap is a well-known open problem in parameterized complexity. Due to our reduction, improved running time lower bounds for finding $k$ balanced biclique may also lead to improved running time lower bounds for VC Dimension and Littlestone's Dimension.

In addition to closing the gap in the time lower bounds, another interesting question--originally posed in~\cite{FrancesL98}--is whether there is an efficient online learner for any given concept class with approximately optimal mistake bound (as defined in \Cref{def:mistake-bound}). Ostensibly, this problem is very similar to that of approximating Littlestone's Dimension. In fact, in the \emph{exact} setting,~\cite{FrancesL98} showed that there is an efficient online learner  with \emph{exactly} optimal mistake bound iff there is an efficient algorithm for \emph{exactly} computing Littlestone's Dimension. However, this reduction breaks down for the approximate setting. In particular, our results do not rule out the fact that an efficient online learner with mistake bound at most, say, twice the optimal exists.  This remains an interesting open question. On this front, we remark that the learner used in our reduction is indeed efficient (given in \Cref{lem:littlestone-from-biclique}); instead, the computational burden falls to the adversary / nature who has to choose the ``hard'' input sequence that induces a large biclique.

\section*{Acknowledgment}

I would like to thank ITCS 2023 reviewers for their helpful comments and suggestions.

\bibliographystyle{alpha}
\bibliography{ref}

\appendix

\section{Parameterized Hardness of Approximation}
\label{app:fpt-inapprox}

In this section, we briefly discuss the implications of our results for parameterized algorithms. Fixed-parameter algorithms (FPT algorithms) are those that run in $f(k) \cdot n^{O(1)}$-time where $k$ denote the parameter (specified as part of the input) and, as is standard, $n$ denote the input size\footnote{For more detail on FPT approximation algorithms and hardness of approximation results, please refer e.g. to the survey~\cite{FeldmannSLM20}.}. In our case, we say that an algorithm is an $\alpha$-approximation FPT algorithm for VC Dimension (resp. Littlestone's Dimension) iff it is an FPT algorithm that can decide between $\VC(\C) \geq k$ and $\VC(\C) < k / \alpha$ (resp. $\Ldim(\C) \geq k$ and $\Ldim(\C) < k / \alpha$). We remark that an $\alpha$-approximation FPT algorithm is only non-trivial iff $\alpha = o(k)$, since there is a trial $n^{k/\alpha}$-time $\alpha$-approximation algorithm. We show that, unfortunately, no non-trivial FPT approximation algorithm exists for VC / Littlestone's Dimensions, as stated below. Previously, W[1]-hardness against \emph{exact} FPT algorithms was known~\cite{DowneyEF93}, but we are not aware of any result ruling out FPT approximation algorithms (even for small approximation ratio e.g. 1.001).

\begin{corollary}
Assuming Gap-ETH, there is no $o(k)$-approximation FPT algorithm for VC Dimension or Littlestone's Dimension.
\end{corollary}

\begin{corollary}
Assuming the Strongish Planted Clique Hypothesis, there is no $o(k)$-approximation FPT algorithm for VC Dimension or Littlestone's Dimension.
\end{corollary}

The above corollaries are direct consequences of the following lemma, which shows that any non-trivial FPT approximation algorithm can be used to obtain $o(\log n)$-approximation in polynomial-time, together with our $o(\log n)$-factor hardness results from \Cref{thm:gap-eth-inapprox} and \Cref{thm:spc-approx}.

\begin{lemma}
If there is an $o(k)$-approximation FPT algorithm for VC Dimension (resp. Littlestone's Dimension), then there also exists a polynomial-time $o(\log n)$-approximation for VC Dimension (resp. Littlestone's Dimension).
\end{lemma}

\begin{proof}
We only prove the statement for VC Dimension; the proof for Littlestone's Dimension is analogous. Suppose that there exists an $\alpha(k)$-approximation algorithm $\A$ that runs in $f(k) \cdot n^{O(1)}$-time for some $\alpha(k) = o(k)$. Let $g:\N \to \N$ be defined as $g(n) := \min\{\max\{k \in \N \mid f(k) \leq n\}, \lfloor\sqrt{\log n}\rfloor\}$. (If the set is empty, let $g(n) = 0$.) Note that $\lim_{n \to \infty} g(n) = \infty$, i.e. $g = \omega(1)$. Let us now consider the following algorithm (where $\C$ denote the input concept class):
\begin{itemize}
\item Run $\A$ on $\C$ with $k = g(n)$.
\item If $\A$ returns YES, then output $k / \alpha(k)$. Otherwise, output 1.
\end{itemize}
By definition of $g(n)$, the algorithm runs in polynomial time. 

To analyze the approximation ratio, consider two cases:
\begin{itemize}
\item $\A$ returns YES. Then, we must have $\VC(\C) \geq k/\alpha(k)$ and we output $k/\alpha(k)$. Therefore, the approximation ratio is at most $\frac{\log n}{k / \alpha(k)} \leq o(\log n)$, where the second inequality follows from $k = g(n) = \omega(1)$ and $\alpha(k) = o(k)$.
\item $\A$ returns NO. Then, we must have $\VC(\C) \leq k$ and we output $1$. The approximation ratio here is at most $k = g(n) \leq \sqrt{\log n}$.
\end{itemize}
Thus, in both cases, the approximation ratio is $o(\log n)$ as desired.
\end{proof}

\section{Proof Sketch of \Cref{lem:pc-biclique} and \Cref{lem:spc-biclique}}
\label{app:pc-biclique}

The reduction from the Planted Clique graph $G = (V, E)$ to an input $G' = (A', B', E')$ to Gap Biclique is as follows\footnote{This reduction is quite standard and is also used e.g. in~\cite{ChalermsookCKLM20}.}:
\begin{itemize}
\item Let $A', B'$ be copies of $V$.
\item Add an edge between $a \in A', b \in B'$ iff $a = b$ or $(a, b)$ belongs to $E$.
\end{itemize}
In the YES case where $G$ contains $\lceil N^\delta \rceil$-clique, then clearly $G'$ also contains $K_{\lceil N^\delta \rceil, \lceil N^\delta \rceil}$. On the other hand, if $G$ is a random $G(N, 1/2)$ graph, a standard union bound argument shows that it does not contain any $K_{\lceil 3 \log N \rceil, \lceil 3 \log N \rceil}$ w.h.p. It is not hard to see (see e.g.~\cite[Lemma 5.17]{ChalermsookCKLM20}) that this implies that $G'$ does not contain $K_{\lceil 6 \log N \rceil, \lceil 6 \log N \rceil}$. This completes the proof sketch.

\end{document}